\documentclass[twocolumn,10pt]{IEEEtran}

\usepackage[acronym]{glossaries}
\newacronym{iid}{i.i.d.}{independent and identically distributed}
\newacronym{siso}{SISO}{single-input single-output}
\newacronym{mimo}{MIMO}{multiple-input multiple-output}
\newacronym{miso}{MISO}{multiple-input single-output}
\newacronym{simo}{SIMO}{single-input multiple-output}
\newacronym{upa}{UPA}{uniform planar array}
\newacronym{em}{EM}{electromagnetic}
\newacronym{adc}{ADC}{analog-to-digital converter}
\newacronym{dac}{DAC}{digital-to-analog converter}
\newacronym{rf}{RF}{radio-frequency}
\newacronym{pin}{PIN}{positive-intrinsic negative}

\usepackage{amssymb}
\usepackage{amsmath}
\usepackage{amsthm}
\usepackage{graphicx}
\usepackage{comment}
\usepackage{xcolor}
\usepackage{subfigure}
\usepackage{cases}

\newtheorem{proposition}{Proposition}

\begin{document}

\title{\huge Physically Consistent Modeling of Stacked Intelligent Metasurfaces Implemented with Beyond Diagonal RIS}

\author{Matteo~Nerini,~\IEEEmembership{Graduate Student Member,~IEEE},
        Bruno~Clerckx,~\IEEEmembership{Fellow,~IEEE}
        
\thanks{M. Nerini and B. Clerckx are with the Department of Electrical and Electronic Engineering, Imperial College London, London SW7 2AZ, U.K. (e-mail: \{m.nerini20, b.clerckx\}@imperial.ac.uk).}}

\maketitle

\begin{abstract}
Stacked intelligent metasurface (SIM) has emerged as a technology enabling wave domain beamforming through multiple stacked reconfigurable intelligent surfaces (RISs).
SIM has been implemented so far with diagonal RIS (D-RIS), while SIM implemented with beyond diagonal RIS (BD-RIS) remains unexplored.
Furthermore, a model of SIM accounting for mutual coupling is not yet available.
To fill these gaps, we derive a physically consistent channel model for SIM-aided systems and clarify the assumptions needed to obtain the simplified model used in related works.
Using this model, we show that 1-layer SIM implemented with BD-RIS achieves the performance upper bound with limited complexity.
\end{abstract}

\glsresetall


%%%%%%%%%%%%%%%%%%%%%%%%%%%%%%%%%%%%%%%%%%%%%%%%%%
\section{Introduction}

% SIM
Recently, a novel transceiver design based on stacked intelligent metasurfaces (SIMs) has been proposed \cite{an23a,an23c}.
A SIM consists of multiple stacked layers, each being a reconfigurable intelligent surface (RIS) allowing the transmission of the incident \gls{em} wave in a controllable manner.
Given the presence of multiple reconfigurable layers, a SIM can efficiently perform beamforming in the \gls{em} wave domain \cite{an23a} and implement holographic \gls{mimo} communications \cite{an23c}.
Specifically, SIM technology is characterized by three main benefits that make it appealing for future wireless communications.
First, SIMs can replace conventional digital beamforming and remove the necessity for high-resolution \gls{dac} and \gls{adc}, reducing the hardware cost.
Second, SIM can reduce the number of needed \gls{rf} chains, consequently decreasing energy consumption.
Third, SIM can remove the latency due to the precoding processing since it is performed in the \gls{em} wave domain.

% No SIM with BD-RIS
In previous works, SIM layers have been implemented using diagonal RIS (D-RIS), which is the conventional RIS architecture characterized by a diagonal phase shift matrix \cite{an23a,an23c}.
However, to overcome the limited flexibility of D-RIS, beyond diagonal RIS (BD-RIS) has been recently proposed as a generalization of D-RIS \cite{li23-1}.
The novelty introduced in BD-RIS is the presence of tunable impedance components interconnecting the RIS elements to each other, enabling the design of RIS with a scattering matrix that is not limited to being diagonal.
Depending on the location of these tunable interconnections, multiple BD-RIS architectures have been proposed, such as fully-/group-connected RISs \cite{she20} and tree-/forest-connected RISs \cite{ner23-1}.
Given the promising performance enabled by BD-RIS, in this study, we model and compare SIMs implemented with D-RIS and BD-RIS.

% No physically consistent models of SIM
To thoroughly compare SIMs built through D-RIS and BD-RIS, a physically consistent model of SIM-aided communication systems is needed.
Despite a path-loss model for SIM-aided systems has been proposed in \cite{has24}, a physically consistent model of SIM accounting for the \gls{em} mutual coupling effects is not yet available.
Nevertheless, \gls{em}-compliant analyses of RIS-aided systems have been developed through multiport network theory using the scattering \cite{she20}, impedance \cite{gra21}, and admittance \cite{ner23-3} parameters.
Furthermore, the relationship between these three analyses has been investigated in \cite{ner23-3,nos23}.
Since multiport network theory has successfully produced physically consistent models for RIS-aided systems, in this study, we employ this theory to model SIM-aided systems.

% Contributions
The contributions of this study are summarized as follows.
\textit{First}, we model a SIM-aided communication system through multiport network theory and derive a general channel model, accounting for the mutual coupling effects at the transmitter, SIM layers, and receiver.
\textit{Second}, we simplify the derived model into a simplified channel model, consistent with the model widely used in recent literature on SIM, clarifying for the first time which are the required assumptions.
\textit{Third}, we analyze and compare SIM architectures implemented with D-RIS and BD-RIS.
Our theoretical analysis corroborated by numerical results shows that 1-layer SIM implemented with BD-RIS achieves the performance upper bound with limited circuit complexity, and that any $L$-layer SIM implemented with D-RIS is suboptimal.

\begin{figure*}[t]
\centering
\includegraphics[width=0.98\textwidth]{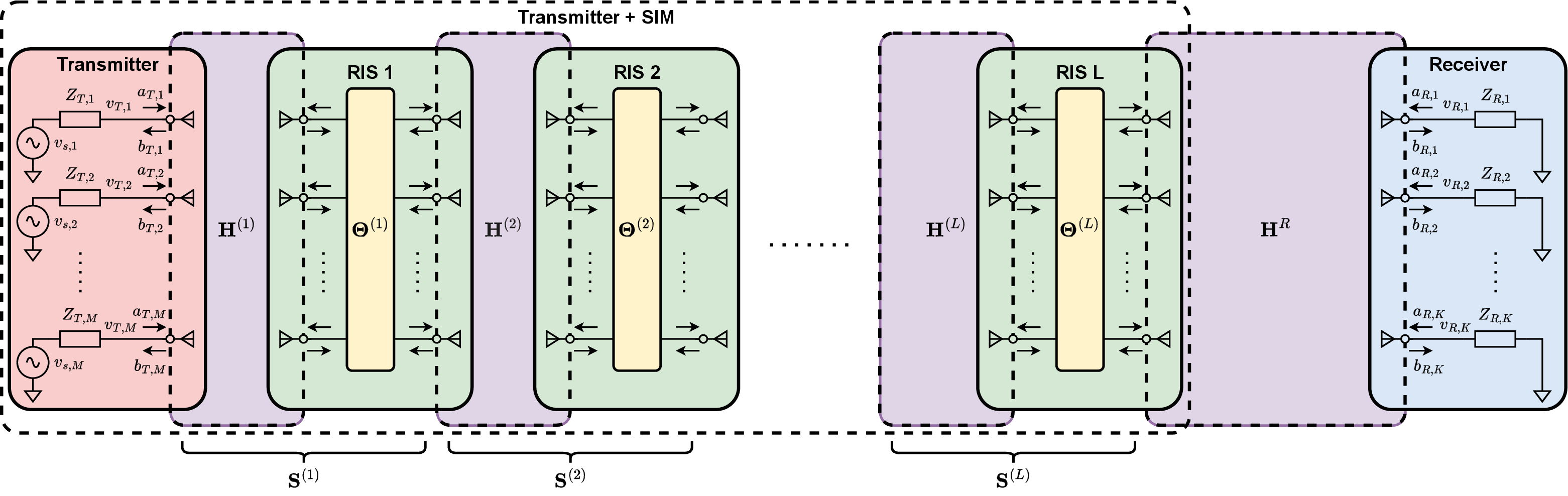}
\caption{SIM-aided communication system diagram.}
\label{fig:SIM}
\end{figure*}

%%%%%%%%%%%%%%%%%%%%%%%%%%%%%%%%%%%%%%%%%%%%%%%%%%
\section{SIM-Aided System Model Based on\\Multiport Network Theory}
\label{sec:model}

Consider a SIM-aided communication system between an $M$-antenna transmitter and a $K$-antenna receiver, as represented in Fig.~\ref{fig:SIM}.
The SIM is deployed at the transmitter and is composed of $L$ layers, each being an $N$-element RIS.
We model this SIM-aided system through multiport network theory as follows \cite[Chapter 4]{poz11}.

% Transmitter
At the transmitter, the $m$th antenna is connected in series with a source voltage $v_{s,m}$ and a source impedance $Z_{T,m}$, for $m=1,\ldots,M$.
Defining $\mathbf{a}_T\in\mathbb{C}^{M\times 1}$ and $\mathbf{b}_T\in\mathbb{C}^{M\times 1}$ as the reflected and incident waves at the transmitter, we have
\begin{equation}
\mathbf{a}_T=\mathbf{b}_{s,T}+\mathbf{\Gamma}_T\mathbf{b}_T,\label{eq:TX}
\end{equation}
where $\mathbf{\Gamma}_{T}\in\mathbb{C}^{M\times M}$ is a diagonal matrix containing the reflection coefficients of the source impedances, i.e.,
\begin{equation}
\mathbf{\Gamma}_T=\left(\mathbf{Z}_T+Z_0\mathbf{I}\right)^{-1}\left(\mathbf{Z}_T-Z_0\mathbf{I}\right),\label{eq:GammaT}
\end{equation}
with $\mathbf{Z}_{T}=\textrm{diag}(Z_{T,1},\ldots,Z_{T,M})$ and $Z_0$ denoting the reference impedance typically set to $Z_0=50$~$\Omega$, and $\mathbf{b}_{s,T}\in\mathbb{C}^{M\times 1}$ is the source wave vector given by $\mathbf{b}_{s,T}=(\mathbf{I}-\mathbf{\Gamma}_T)\mathbf{v}_{s,T}/2$, with $\mathbf{v}_{s,T}=[v_{s,1},\ldots,v_{s,M}]^T$ \cite{ner23-3}.

% SIM: 2 - L layers
The SIM is a cascade of $L$ blocks, each consisting of the cascade of a wireless channel and an RIS.
Specifically, the $\ell$th block consists of the cascade of the wireless channel between the $(\ell-1)$th and the $\ell$th SIM layer and the RIS implementing the $\ell$th SIM layer, for $\ell=2,\ldots,L$.
We model the channel between the $(\ell-1)$th and the $\ell$th SIM layer as a $2N$-port network with scattering matrix $\mathbf{H}^{(\ell)}\in\mathbb{C}^{2N\times 2N}$ given by
\begin{equation}
\mathbf{H}^{(\ell)}=
\begin{bmatrix}
\mathbf{H}_{11}^{(\ell)} & \mathbf{H}_{12}^{(\ell)}\\
\mathbf{H}_{21}^{(\ell)} & \mathbf{H}_{22}^{(\ell)}
\end{bmatrix},\label{eq:HL}
\end{equation}
where $\mathbf{H}_{11}^{(\ell)}\in\mathbb{C}^{N\times N}$ and $\mathbf{H}_{22}^{(\ell)}\in\mathbb{C}^{N\times N}$ refer to the antenna mutual coupling and self-impedance at the $(\ell-1)$ and $\ell$th SIM layer, respectively, and $\mathbf{H}_{21}^{(\ell)}\in\mathbb{C}^{N\times N}$ refer to the transmission scattering matrix from the $(\ell-1)$th to the $\ell$th SIM layer, for $\ell=2,\ldots,L$.
Beside, we model the RIS at the $\ell$th SIM layer through its scattering matrix $\boldsymbol{\Theta}^{(\ell)}\in\mathbb{C}^{2N\times 2N}$ given by
\begin{equation}
\boldsymbol{\Theta}^{(\ell)}=
\begin{bmatrix}
\boldsymbol{\Theta}_{11}^{(\ell)} & \boldsymbol{\Theta}_{12}^{(\ell)}\\
\boldsymbol{\Theta}_{21}^{(\ell)} & \boldsymbol{\Theta}_{22}^{(\ell)}
\end{bmatrix},\label{eq:TL}
\end{equation}
for $\ell=2,\ldots,L$, which is symmetric in the case of a reciprocal RIS and whose structure depends on the RIS architecture.
In the case the SIM layers are implemented through conventional D-RISs, we have $\boldsymbol{\Theta}_{11}^{(\ell)}=\mathbf{0}$ and $\boldsymbol{\Theta}_{22}^{(\ell)}=\mathbf{0}$ since the RISs work in transmissive mode \cite{li22-1}, and
\begin{equation}
\boldsymbol{\Theta}_{12}^{(\ell)}=\boldsymbol{\Theta}_{21}^{(\ell)}=\text{diag}\left(e^{j\theta_{1}^{(\ell)}},\ldots,e^{j\theta_{N}^{(\ell)}}\right),\label{eq:T21}
\end{equation}
assuming the RISs to be lossless \cite{an23a,an23c}.
However, these constraints valid for D-RISs can be relaxed into the more general constraint $\boldsymbol{\Theta}^{(\ell)H}\boldsymbol{\Theta}^{(\ell)}=\mathbf{I}$ by implementing the SIM layers through BD-RISs \cite{li22-1}.

% SIM: 1 layer
Note that the first block, i.e., $\ell=1$, consists of the cascade of the wireless channel between the transmitter and the first SIM layer and the RIS implementing the first SIM layer.
Thus, the scattering matrix of the wireless channel of the first block $\mathbf{H}^{(1)}\in\mathbb{C}^{(M+N)\times (M+N)}$ is given as in \eqref{eq:HL}, where $\mathbf{H}_{11}^{(1)}\in\mathbb{C}^{M\times M}$ refer to the antenna mutual coupling and self-impedance at the transmitter.
Besides the scattering matrix of the first SIM layer $\boldsymbol{\Theta}^{(1)}$ is given as in \eqref{eq:TL}.

% Wireless channel
The wireless channel between the $L$th SIM layer and the receiver is modeled as an $(N+K)$-port network with scattering matrix $\mathbf{H}^{R}\in\mathbb{C}^{(N+K)\times (N+K)}$ given by
\begin{equation}
\mathbf{H}^{R}=
\begin{bmatrix}
\mathbf{H}_{11}^{R} & \mathbf{H}_{12}^{R}\\
\mathbf{H}_{21}^{R} & \mathbf{H}_{22}^{R}
\end{bmatrix},
\end{equation}
where $\mathbf{H}_{11}^{R}\in\mathbb{C}^{N\times N}$ and $\mathbf{H}_{22}^{R}\in\mathbb{C}^{K\times K}$ refer to the antenna mutual coupling and self-impedance at the $L$th SIM layer and the receiver, and $\mathbf{H}_{21}^{R}\in\mathbb{C}^{K\times N}$ refer to the transmission scattering matrix from the $L$th SIM layer to the receiver.

% Receiver
At the receiver, the $k$th antenna is connected in series with a load impedance $Z_{R,k}$, for $k=1,\ldots,K$.
Defining $\mathbf{a}_R\in\mathbb{C}^{K\times 1}$ and $\mathbf{b}_R\in\mathbb{C}^{K\times 1}$ the reflected and incident waves at the receiver, we have
\begin{equation}
\mathbf{a}_R=\mathbf{\Gamma}_R\mathbf{b}_R,\label{eq:RX}
\end{equation}
where $\mathbf{\Gamma}_{R}\in\mathbb{C}^{K\times K}$ is a diagonal matrix containing the reflection coefficients of the load impedances, i.e.,
\begin{equation}
\mathbf{\Gamma}_R=\left(\mathbf{Z}_R+Z_0\mathbf{I}\right)^{-1}\left(\mathbf{Z}_R-Z_0\mathbf{I}\right),\label{eq:GammaR}
\end{equation}
with $\mathbf{Z}_{R}=\textrm{diag}(Z_{R,1},\ldots,Z_{R,K})$ \cite{ner23-3}.

%%%%%%%%%%%%%%%%%%%%%%%%%%%%%%%%%%%%%%%%%%%%%%%%%%
\section{General SIM-Aided Channel Model}

Our goal is to determine the expression of the channel matrix $\mathbf{H}\in\mathbb{C}^{K\times M}$ relating the voltage vector at the transmitter $\mathbf{v}_T\in\mathbb{C}^{M\times 1}$ and the voltage vector at the receiver $\mathbf{v}_R\in\mathbb{C}^{K\times 1}$ though $\mathbf{v}_R=\mathbf{H}\mathbf{v}_T$, where $\mathbf{v}_T$ and $\mathbf{v}_R$ are given by
\begin{equation}
\mathbf{v}_T=\mathbf{a}_T+\mathbf{b}_T,\:\mathbf{v}_R=\mathbf{a}_R+\mathbf{b}_R,\label{eq:vtvr}
\end{equation}
according to multiport network theory \cite[Chapter 4]{poz11}.
To this end, we introduce a proposition enabling the analysis of the considered SIM-aided system.
\begin{proposition}
Consider the cascade system represented in Fig.~\ref{fig:cascade}, consisting of an $N_P$-port network, with scattering matrix $\mathbf{P}\in\mathbb{C}^{N_P\times N_P}$ given by
\begin{equation}
\mathbf{P}=
\begin{bmatrix}
\mathbf{P}_{11} & \mathbf{P}_{12}\\
\mathbf{P}_{21} & \mathbf{P}_{22}
\end{bmatrix},
\end{equation}
where $\mathbf{P}_{11}\in\mathbb{C}^{N_1\times N_1}$ and $\mathbf{P}_{22}\in\mathbb{C}^{N_2\times N_2}$, with $N_P=N_1+N_2$, and an $N_Q$-port network, with scattering matrix $\mathbf{Q}\in\mathbb{C}^{N_Q\times N_Q}$ given by
\begin{equation}
\mathbf{Q}=
\begin{bmatrix}
\mathbf{Q}_{11} & \mathbf{Q}_{12}\\
\mathbf{Q}_{21} & \mathbf{Q}_{22}
\end{bmatrix},
\end{equation}
where $\mathbf{Q}_{11}\in\mathbb{C}^{N_2\times N_2}$ and $\mathbf{Q}_{22}\in\mathbb{C}^{N_3\times N_3}$, with $N_Q=N_2+N_3$.
In this cascade, the last $N_2$ ports of the $N_P$-port network are individually connected to the first $N_2$ ports of the $N_Q$-port network.
Thus, the whole cascade can be regarded as an $N_R$-port network, where $N_R=N_1+N_3$, with scattering matrix $\mathbf{R}\in\mathbb{C}^{N_R\times N_R}$ given by
\begin{equation}
\mathbf{R}=
\begin{bmatrix}
\mathbf{R}_{11} & \mathbf{R}_{12}\\
\mathbf{R}_{21} & \mathbf{R}_{22}
\end{bmatrix},
\end{equation}
with
\begin{equation}
\mathbf{R}_{11}=\mathbf{P}_{11}+\mathbf{P}_{12}\left(\mathbf{I}-\mathbf{Q}_{11}\mathbf{P}_{22}\right)^{-1}\mathbf{Q}_{11}\mathbf{P}_{21},
\end{equation}
\begin{equation}
\mathbf{R}_{12}=\mathbf{P}_{12}\left(\mathbf{I}-\mathbf{Q}_{11}\mathbf{P}_{22}\right)^{-1}\mathbf{Q}_{12},
\end{equation}
\begin{equation}
\mathbf{R}_{21}=\mathbf{Q}_{21}\left(\mathbf{I}-\mathbf{P}_{22}\mathbf{Q}_{11}\right)^{-1}\mathbf{P}_{21},
\end{equation}
\begin{equation}
\mathbf{R}_{22}=\mathbf{Q}_{22}+\mathbf{Q}_{21}\left(\mathbf{I}-\mathbf{P}_{22}\mathbf{Q}_{11}\right)^{-1}\mathbf{P}_{22}\mathbf{Q}_{12}.
\end{equation}
\label{pro:cascade}
\end{proposition}
\begin{proof}
Please refer to the Appendix.
\end{proof}

\begin{figure}[t]
\centering
\includegraphics[width=0.32\textwidth]{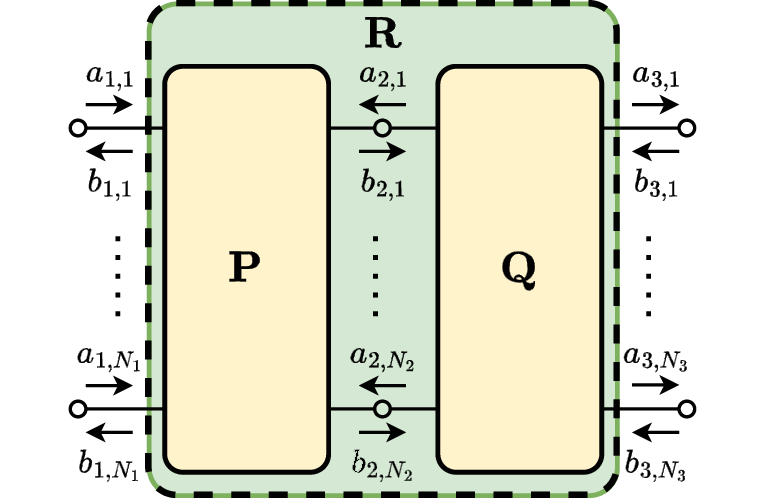}
\caption{Cascade of two multiport networks.}
\label{fig:cascade}
\end{figure}

To derive the expression of $\mathbf{H}$, we first recall that the $\ell$th block of the SIM is formed by the cascade of the channel $\mathbf{H}^{(\ell)}$ and the RIS $\boldsymbol{\Theta}^{(\ell)}$, given by \eqref{eq:HL} and \eqref{eq:TL}, respectively.
Thus, the scattering matrix of the $\ell$th block is given by
\begin{equation}
\mathbf{S}^{(\ell)}=
\begin{bmatrix}
\mathbf{S}_{11}^{(\ell)} & \mathbf{S}_{12}^{(\ell)}\\
\mathbf{S}_{21}^{(\ell)} & \mathbf{S}_{22}^{(\ell)}
\end{bmatrix},\label{eq:SL}
\end{equation}
with
\begin{equation}
\mathbf{S}_{11}^{(\ell)}=\mathbf{H}_{11}^{(\ell)}+\mathbf{H}_{12}^{(\ell)}\left(\mathbf{I}-\boldsymbol{\Theta}_{11}^{(\ell)}\mathbf{H}_{22}^{(\ell)}\right)^{-1}\boldsymbol{\Theta}_{11}^{(\ell)}\mathbf{H}_{21}^{(\ell)},
\end{equation}
\begin{equation}
\mathbf{S}_{12}^{(\ell)}=\mathbf{H}_{12}^{(\ell)}\left(\mathbf{I}-\boldsymbol{\Theta}_{11}^{(\ell)}\mathbf{H}_{22}^{(\ell)}\right)^{-1}\boldsymbol{\Theta}_{12}^{(\ell)},
\end{equation}
\begin{equation}
\mathbf{S}_{21}^{(\ell)}=\boldsymbol{\Theta}_{21}^{(\ell)}\left(\mathbf{I}-\mathbf{H}_{22}^{(\ell)}\boldsymbol{\Theta}_{11}^{(\ell)}\right)^{-1}\mathbf{H}_{21}^{(\ell)},
\end{equation}
\begin{equation}
\mathbf{S}_{22}^{(\ell)}=\boldsymbol{\Theta}_{22}^{(\ell)}+\boldsymbol{\Theta}_{21}^{(\ell)}\left(\mathbf{I}-\mathbf{H}_{22}^{(\ell)}\boldsymbol{\Theta}_{11}^{(\ell)}\right)^{-1}\mathbf{H}_{22}^{(\ell)}\boldsymbol{\Theta}_{12}^{(\ell)},
\end{equation}
for $\ell=1,\ldots,L$, following Proposition~\ref{pro:cascade}.
Given $\mathbf{S}^{(\ell)}$, the equivalent scattering matrix including the first $\ell$ blocks of the SIM can be determined recursively by
\begin{equation}
\bar{\mathbf{S}}^{(\ell)}=
\begin{bmatrix}
\bar{\mathbf{S}}_{11}^{(\ell)} & \bar{\mathbf{S}}_{12}^{(\ell)}\\
\bar{\mathbf{S}}_{21}^{(\ell)} & \bar{\mathbf{S}}_{22}^{(\ell)}
\end{bmatrix},
\end{equation}
with
\begin{equation}
\bar{\mathbf{S}}_{11}^{(\ell)}=\bar{\mathbf{S}}_{11}^{(\ell-1)}+\bar{\mathbf{S}}_{12}^{(\ell-1)}\left(\mathbf{I}-\mathbf{S}_{11}^{(\ell)}\bar{\mathbf{S}}_{22}^{(\ell-1)}\right)^{-1}\mathbf{S}_{11}^{(\ell)}\bar{\mathbf{S}}_{21}^{(\ell-1)},
\end{equation}
\begin{equation}
\bar{\mathbf{S}}_{12}^{(\ell)}=\bar{\mathbf{S}}_{12}^{(\ell-1)}\left(\mathbf{I}-\mathbf{S}_{11}^{(\ell)}\bar{\mathbf{S}}_{22}^{(\ell-1)}\right)^{-1}\mathbf{S}_{12}^{(\ell)},
\end{equation}
\begin{equation}
\bar{\mathbf{S}}_{21}^{(\ell)}=\mathbf{S}_{21}^{(\ell)}\left(\mathbf{I}-\bar{\mathbf{S}}_{22}^{(\ell-1)}\mathbf{S}_{11}^{(\ell)}\right)^{-1}\bar{\mathbf{S}}_{21}^{(\ell-1)},
\end{equation}
\begin{equation}
\bar{\mathbf{S}}_{22}^{(\ell)}=\mathbf{S}_{22}^{(\ell)}+\mathbf{S}_{21}^{(\ell)}\left(\mathbf{I}-\bar{\mathbf{S}}_{22}^{(\ell-1)}\mathbf{S}_{11}^{(\ell)}\right)^{-1}\bar{\mathbf{S}}_{22}^{(\ell-1)}\mathbf{S}_{12}^{(\ell)},
\end{equation}
as a function of $\bar{\mathbf{S}}^{(\ell-1)}$ and $\mathbf{S}^{(\ell)}$, where $\bar{\mathbf{S}}^{(1)}=\mathbf{S}^{(1)}$, according to Proposition~\ref{pro:cascade}.
In particular, $\bar{\mathbf{S}}^{(L)}$ represents the equivalent scattering matrix including all the $L$ blocks of the SIM.
From $\bar{\mathbf{S}}^{(L)}$ and the channel $\mathbf{H}^{R}$, the scattering matrix of the entire channel between the transmitter and the receiver is given by the cascade of $\bar{\mathbf{S}}^{(L)}$ and $\mathbf{H}^{R}$, and writes as
\begin{equation}
\mathbf{S}=
\begin{bmatrix}
\mathbf{S}_{11} & \mathbf{S}_{12}\\
\mathbf{S}_{21} & \mathbf{S}_{22}
\end{bmatrix},
\end{equation}
with
\begin{equation}
\mathbf{S}_{11}=\bar{\mathbf{S}}_{11}^{(L)}+\bar{\mathbf{S}}_{12}^{(L)}\left(\mathbf{I}-\mathbf{H}_{11}^{R}\bar{\mathbf{S}}_{22}^{(L)}\right)^{-1}\mathbf{H}_{11}^{R}\bar{\mathbf{S}}_{21}^{(L)},
\end{equation}
\begin{equation}
\mathbf{S}_{12}=\bar{\mathbf{S}}_{12}^{(L)}\left(\mathbf{I}-\mathbf{H}_{11}^{R}\bar{\mathbf{S}}_{22}^{(L)}\right)^{-1}\mathbf{H}_{12}^{R},
\end{equation}
\begin{equation}
\mathbf{S}_{21}=\mathbf{H}_{21}^{R}\left(\mathbf{I}-\bar{\mathbf{S}}_{22}^{(L)}\mathbf{H}_{11}^{R}\right)^{-1}\bar{\mathbf{S}}_{21}^{(L)},
\end{equation}
\begin{equation}
\mathbf{S}_{22}=\mathbf{H}_{22}^{R}+\mathbf{H}_{21}^{R}\left(\mathbf{I}-\bar{\mathbf{S}}_{22}^{(L)}\mathbf{H}_{11}^{R}\right)^{-1}\bar{\mathbf{S}}_{22}^{(L)}\mathbf{H}_{12}^{R},
\end{equation}
according to Proposition~\ref{pro:cascade}.

To express $\mathbf{H}$ as a function of $\mathbf{S}$, we consider $\mathbf{Z}_T=Z_0\mathbf{I}$ and $\mathbf{Z}_R=Z_0\mathbf{I}$ for simplicity, so that $\mathbf{\Gamma}_T=\mathbf{0}$ and $\mathbf{\Gamma}_R=\mathbf{0}$ according to \eqref{eq:GammaT} and \eqref{eq:GammaR}, respectively.
Consequently, $\mathbf{H}$ can be derived by solving the system
\begin{numcases}{}
\begin{bmatrix}
\mathbf{b}_T\\
\mathbf{b}_R
\end{bmatrix}=
\begin{bmatrix}
\mathbf{S}_{11} & \mathbf{S}_{12}\\
\mathbf{S}_{21} & \mathbf{S}_{22}
\end{bmatrix}
\begin{bmatrix}
\mathbf{a}_T\\
\mathbf{a}_R
\end{bmatrix},\label{eq:sis1}\\
\mathbf{a}_T=\mathbf{b}_{s,T},\:\mathbf{a}_R=\mathbf{0},\label{eq:sis2}\\
\eqref{eq:vtvr},\label{eq:sis3}
\end{numcases}
where \eqref{eq:sis1} follows the definition of scattering matrix \cite[Chapter 4]{poz11} and \eqref{eq:sis2} is given by \eqref{eq:TX} and \eqref{eq:RX}.
%\begin{align}
%\mathbf{v}_R
%&=\mathbf{a}_R+\mathbf{b}_R\\
%&=\mathbf{b}_R\\
%&=\mathbf{S}_{21}\mathbf{b}_{s,T}
%\end{align}
%\begin{align}
%\mathbf{v}_T
%&=\mathbf{a}_T+\mathbf{b}_T\\
%&=\mathbf{b}_{s,T}+\mathbf{S}_{11}\mathbf{b}_{s,T}\\
%&=\left(\mathbf{I}+\mathbf{S}_{11}\right)\mathbf{b}_{s,T}
%\end{align}
Since it is possible to show that system \eqref{eq:sis1}-\eqref{eq:sis3} yields $\mathbf{v}_R=\mathbf{S}_{21}\left(\mathbf{I}+\mathbf{S}_{11}\right)^{-1}\mathbf{v}_T$, we have
\begin{equation}
\mathbf{H}=\mathbf{S}_{21}\left(\mathbf{I}+\mathbf{S}_{11}\right)^{-1},\label{eq:Hgen}
\end{equation}
giving the general expression of a SIM-aided channel.
Remarkably, it is hard to understand the role of the reconfigurable scattering matrices $\boldsymbol{\Theta}^{(\ell)}$ in \eqref{eq:Hgen}.
For this reason, we simplify this channel model and show under which assumptions it boils down to the easier model widely used in recent works on SIM.

%%%%%%%%%%%%%%%%%%%%%%%%%%%%%%%%%%%%%%%%%%%%%%%%%%
\section{Simplified SIM-Aided Channel Model}

To simplify the channel model in \eqref{eq:Hgen}, we assume that the unilateral approximation is valid in all the channels $\mathbf{H}^{(\ell)}$, i.e., $\mathbf{H}_{12}^{(\ell)}=\mathbf{0}$, for $\ell=1,\ldots,L$ \cite{ivr10}.
In other words, we assume that the electrical properties at the $(\ell-1)$th SIM layer are independent of the electrical properties at the $\ell$th SIM layer (where the $0$th SIM layer is the transmitter).
In addition, we assume no mutual coupling between the transmit and RIS antennas, and perfect matching to $Z_0$, i.e., $\mathbf{H}_{11}^{(\ell)}=\mathbf{0}$ and $\mathbf{H}_{22}^{(\ell)}=\mathbf{0}$, for $\ell=1,\ldots,L$, and $\mathbf{H}_{11}^{R}=\mathbf{0}$ \cite{ner23-3}.
Given the short transmission distances in a SIM, extensive electromagnetic simulations and real-world measurements are needed to verify the validity of these assumptions, which is beyond the scope of this study.
With these assumptions, the scattering matrix of the $\ell$th block of the SIM given by \eqref{eq:SL} simplifies as
\begin{equation}
\mathbf{S}^{(\ell)}=
\begin{bmatrix}
\mathbf{0} & \mathbf{0}\\
\boldsymbol{\Theta}_{21}^{(\ell)}\mathbf{H}_{21}^{(\ell)} & \boldsymbol{\Theta}_{22}^{(\ell)}
\end{bmatrix},
\end{equation}
for $\ell=1,\ldots,L$.
Thus, the equivalent scattering matrix of the first $\ell$ blocks of the SIM is given recursively as
\begin{equation}
\bar{\mathbf{S}}^{(\ell)}=
\begin{bmatrix}
\bar{\mathbf{S}}_{11}^{(\ell-1)} & \mathbf{0}\\
\mathbf{S}_{21}^{(\ell)}\bar{\mathbf{S}}_{21}^{(\ell-1)} & \mathbf{S}_{22}^{(\ell)}
\end{bmatrix},
\end{equation}
where $\bar{\mathbf{S}}^{(1)}=\mathbf{S}^{(1)}$.
Consequently, the equivalent scattering matrix including all the $L$ blocks of the SIM can now be written in closed-form as
\begin{equation}
\bar{\mathbf{S}}^{(L)}=
\begin{bmatrix}
\mathbf{0} & \mathbf{0}\\
\boldsymbol{\Theta}_{21}^{(L)}\mathbf{H}_{21}^{(L)}\cdots\boldsymbol{\Theta}_{21}^{(2)}\mathbf{H}_{21}^{(2)}\boldsymbol{\Theta}_{21}^{(1)}\mathbf{H}_{21}^{(1)} & \boldsymbol{\Theta}_{22}^{(L)}
\end{bmatrix}.
\end{equation}
Given $\bar{\mathbf{S}}^{(L)}$ and $\mathbf{H}^{R}$, the scattering matrix of the entire channel between the transmitter and the receiver writes as
\begin{equation}
\mathbf{S}=
\begin{bmatrix}
\mathbf{0} & \mathbf{0}\\
\mathbf{H}_{21}^{R}\boldsymbol{\Theta}_{21}^{(L)}\mathbf{H}_{21}^{(L)}\cdots\boldsymbol{\Theta}_{21}^{(1)}\mathbf{H}_{21}^{(1)} & \mathbf{H}_{22}^{R}+\mathbf{H}_{21}^{R}\boldsymbol{\Theta}_{22}^{(L)}\mathbf{H}_{12}^{R}
\end{bmatrix},\label{eq:S}
\end{equation}
by Proposition~\ref{pro:cascade}.
By substituting \eqref{eq:S} into \eqref{eq:Hgen}, and using the auxiliary notation $\bar{\mathbf{H}}^{R}=\mathbf{H}_{21}^{R}$, $\bar{\boldsymbol{\Theta}}^{(\ell)}=\boldsymbol{\Theta}_{21}^{(\ell)}$, and $\bar{\mathbf{H}}^{(\ell)}=\mathbf{H}_{21}^{(\ell)}$, for $\ell=1,\ldots,L$, we finally obtain
\begin{equation}
\mathbf{H}=\bar{\mathbf{H}}^{R}\bar{\boldsymbol{\Theta}}^{(L)}\bar{\mathbf{H}}^{(L)}\cdots\bar{\boldsymbol{\Theta}}^{(2)}\bar{\mathbf{H}}^{(2)}\bar{\boldsymbol{\Theta}}^{(1)}\bar{\mathbf{H}}^{(1)},\label{eq:Hsim}
\end{equation}
which is the simplified channel model, in agreement with the model widely utilized in recent works \cite{an23a,an23c,has24}.

%%%%%%%%%%%%%%%%%%%%%%%%%%%%%%%%%%%%%%%%%%%%%%%%%%
\section{Implementing SIM with D-RIS and BD-RIS}

Consider a \gls{siso} system, i.e., $M=1$ and $K=1$, where we maximize the channel gain
\begin{equation}
\left\vert h\right\vert^2=\left\vert\bar{\mathbf{h}}^{R}\bar{\boldsymbol{\Theta}}^{(L)}\bar{\mathbf{H}}^{(L)}\cdots\bar{\boldsymbol{\Theta}}^{(2)}\bar{\mathbf{H}}^{(2)}\bar{\boldsymbol{\Theta}}^{(1)}\bar{\mathbf{h}}^{(1)}\right\vert^2,\label{eq:gain}
\end{equation}
by optimizing the SIM, i.e., $\bar{\boldsymbol{\Theta}}^{(\ell)}$, for $\ell=1,\ldots,L$.

% D-RIS
In the case of a SIM implemented with D-RIS (the only example of SIM available in literature), $\bar{\boldsymbol{\Theta}}^{(\ell)}$ are subject to \eqref{eq:T21} and a generous upper bound on the achievable gain is
\begin{equation}
\left\vert h^{\text{D-RIS}}\right\vert^2
\leq\left\Vert\bar{\mathbf{h}}^{R}\right\Vert^2\left\Vert\bar{\mathbf{H}}^{(L)}\right\Vert^2\cdots\left\Vert\bar{\mathbf{H}}^{(2)}\right\Vert^2\left\Vert\bar{\mathbf{h}}^{(1)}\right\Vert^2,\label{eq:DRIS-UB}
\end{equation}
following the submultiplicativity property of the spectral norm and noticing that $\Vert\bar{\boldsymbol{\Theta}}^{(\ell)}\Vert=1$, for $\ell=1,\ldots,L$.
To optimize the SIM, we propose an iterative algorithm that optimizes a matrix $\bar{\boldsymbol{\Theta}}^{(\ell)}$ at each iteration through a closed-form globally optimal solution.
Specifically, at each iteration, $\bar{\boldsymbol{\Theta}}^{(\ell)}$ is updated by introducing $\mathbf{u}^{(\ell)}=\bar{\mathbf{h}}^{R}\bar{\boldsymbol{\Theta}}^{(L)}\bar{\mathbf{H}}^{(L)}\cdots\bar{\boldsymbol{\Theta}}^{(\ell+1)}\bar{\mathbf{H}}^{(\ell+1)}$ and $\mathbf{v}^{(\ell)}=\bar{\mathbf{H}}^{(\ell)}\bar{\boldsymbol{\Theta}}^{(\ell-1)}\bar{\mathbf{H}}^{(\ell-1)}\cdots\bar{\boldsymbol{\Theta}}^{(1)}\bar{\mathbf{h}}^{(1)}$, and setting $\theta_n^{(\ell)}=-\mathrm{arg}([\mathbf{u}^{(\ell)}]_n[\mathbf{v}^{(\ell)}]_n)$,
%\begin{equation}
%\theta_n^{(\ell)}=-\mathrm{arg}\left(\left[\mathbf{u}^{(\ell)}\right]_n\left[\mathbf{v}^{(\ell)}\right]_n\right),
%\end{equation}
for $n=1,\ldots,N$.
The algorithm iterates until convergence of the objective \eqref{eq:gain}.
In terms of circuit complexity, when the SIM layers are implemented through D-RISs working in transmissive mode, the number of tunable impedances in the circuit topology is $3N$ per layer \cite{li22-1}, yielding $3NL$ total tunable impedances in the SIM architecture.

% BD-RIS
In the case of a SIM implemented with BD-RIS, we first consider $L=1$ and we show that this is the optimal number of layers.
When $L=1$, \eqref{eq:gain} boils down to $\vert h\vert^2=\vert\bar{\mathbf{h}}^{R}\bar{\boldsymbol{\Theta}}^{(1)}\bar{\mathbf{h}}^{(1)}\vert^2$, which can be rewritten as $\vert h\vert^2=\vert[\mathbf{0},\bar{\mathbf{h}}^{R}]\boldsymbol{\Theta}^{(1)}[\bar{\mathbf{h}}^{(1)T},\mathbf{0}]^T\vert^2$,
where $\mathbf{0}$ denotes here a $1\times N$ zero vector and $\boldsymbol{\Theta}^{(1)}$ is subject to $\boldsymbol{\Theta}^{(1)}=\boldsymbol{\Theta}^{(1)T}$ and $\boldsymbol{\Theta}^{(1)H}\boldsymbol{\Theta}^{(1)}=\mathbf{I}$.
Interestingly, it has been shown in \cite{ner22} that $\boldsymbol{\Theta}^{(1)}$ can be optimized to obtain
\begin{equation}
\left\vert h^{\text{BD-RIS}}\right\vert^2
=\left\Vert\bar{\mathbf{h}}^{R}\right\Vert^2\left\Vert\bar{\mathbf{h}}^{(1)}\right\Vert^2.\label{eq:BDRIS-UB}
\end{equation}
Besides, the performance upper bound in \eqref{eq:BDRIS-UB} can also be achieved with low-complexity BD-RIS architectures denoted as tree-connected RISs through a closed-form solution \cite{ner23-1}.
Since it holds that $\Vert\bar{\mathbf{H}}^{(\ell)}\Vert<1$ for the wireless channels $\bar{\mathbf{H}}^{(\ell)}$, the upper bound in \eqref{eq:BDRIS-UB} is always greater than \eqref{eq:DRIS-UB}.
Thus, $L=1$ is the optimal number of layers in a SIM implemented with tree-connected RISs.
Furthermore, we have $\left\vert h^{\text{D-RIS}}\right\vert^2\leq\left\vert h^{\text{BD-RIS}}\right\vert^2$ as a consequence of \eqref{eq:DRIS-UB} and \eqref{eq:BDRIS-UB}.
Considering a 1-layer SIM implemented through a tree-connected RIS, the number of needed tunable impedances is $4N-1$ \cite{ner23-1}.
Note that prototyping SIM based on BD-RIS poses the same challenges as implementing BD-RIS with transmissive capabilities \cite{li22-1}.

%%%%%%%%%%%%%%%%%%%%%%%%%%%%%%%%%%%%%%%%%%%%%%%%%%
\section{Numerical Results}

% Setup
Consider a 3D system where we locate a \gls{siso} system between a SIM-aided transmitter and a receiver\footnote{Note that the effective number of antennas transmitting the signal is $N$, i.e., the number of RIS elements of the $L$th SIM layer.}.
The transmit antenna is located at $(0,0,0)$ and the SIM layers are $N_x\times N_y$ \glspl{upa} parallel to the $x$-$y$ plane and centered in $(0,0)$.
The SIM layers are spaced $\lambda$ from each other and from the transmit antenna, where $\lambda$ denotes the wavelength.
The RIS elements in each SIM layer have inter-element distance $\lambda/2$ and we set $N_x=4$ and carrier frequency 28~GHz.
The channels $\bar{\mathbf{H}}^{(\ell)}$ are given by diffraction theory as
\begin{equation}
\left[\bar{\mathbf{H}}^{(\ell)}\right]_{i,j}=\frac{A\cos\left(\alpha_{i,j}^{(\ell)}\right)}{d_{i,j}^{(\ell)}}\left(\frac{1}{2\pi d_{i,j}^{(\ell)}}-j\frac{1}{\lambda}\right)e^{j2\pi d_{i,j}^{(\ell)}/\lambda},
\end{equation}
where $A=(\lambda/2)^2$ is the RIS element area, $\alpha_{i,j}^{(\ell)}$ is the angle between the normal of the SIM layers and the propagation direction and $d_{i,j}^{(\ell)}$ is the transmission distance \cite{an23a}.
$\bar{\mathbf{h}}^{R}$ is \gls{iid} Rayleigh distributed.

% Figure performance
In Fig.~\ref{fig:results}, we report the performance in terms of normalized channel gain, defined as $G=\vert h\vert^2/(\Vert\bar{\mathbf{h}}^{R}\Vert^2\Vert\bar{\mathbf{h}}^{(1)}\Vert^2)$, achieved by SIM implemented with D-RIS and BD-RIS.
For SIM implemented with BD-RIS, we consider a 1-layer SIM implemented through a tree-connected RIS, optimized as proposed in \cite{ner23-1}.
We make the following observations.
\textit{First}, 1-layer SIMs implemented with BD-RIS achieve the upper bound $G=1$, as expected from \eqref{eq:BDRIS-UB}.
\textit{Second}, SIMs implemented with D-RIS always underperform 1-layer SIMs implemented with BD-RIS, in agreement with our theoretical derivations.
\textit{Third}, for a low number of RIS elements $N$, SIMs with a low number of layers $L$ are beneficial.
This is because a low $N$ reduces the norms $\Vert\bar{\mathbf{H}}^{(\ell)}\Vert$, for $\ell=1,\ldots,L$, which are always $\Vert\bar{\mathbf{H}}^{(\ell)}\Vert<1$.
Consequently, when $\Vert\bar{\mathbf{H}}^{(\ell)}\Vert<<1$, a high $L$ is suboptimal as it limits the product in the performance upper bound \eqref{eq:DRIS-UB}.
\textit{Fourth}, for a high value of $N$, SIMs with more layers are beneficial given their enhanced flexibility.

% Figure complexity
In Fig.~\ref{fig:results}, we also report the circuit complexity, in terms of the number of tunable impedance components in the circuit topology, also representative of the energy consumption of the SIM.
We observe that 1-layer SIMs implemented with BD-RIS are also favorable in terms of complexity, as they are slightly more complex than 1-layer SIMs implemented with D-RIS.

\begin{figure}[t]
\centering
\includegraphics[width=0.24\textwidth]{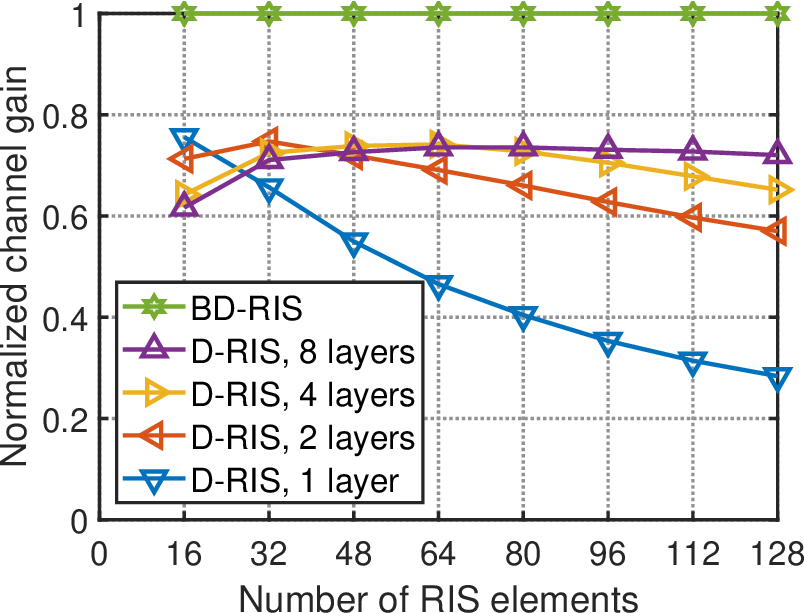}
\includegraphics[width=0.24\textwidth]{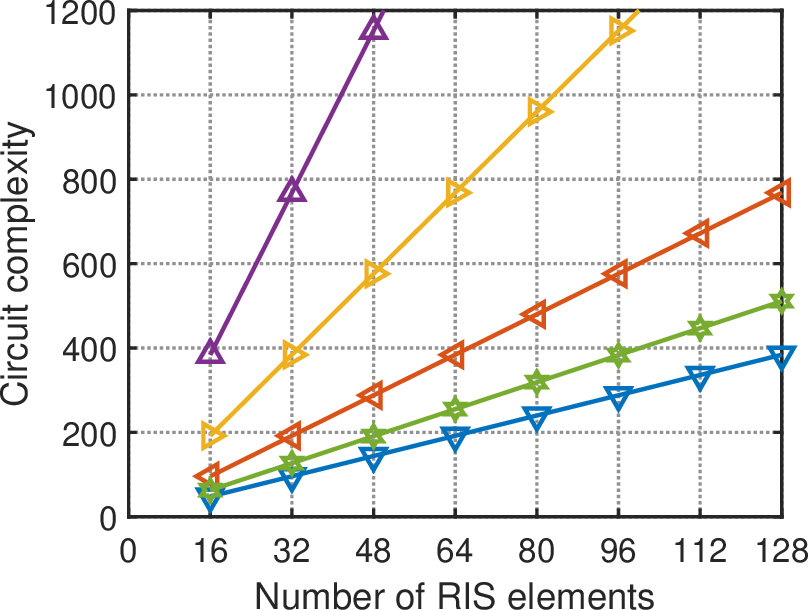}
\caption{Normalized channel gain $G=\vert h\vert^2/(\Vert\bar{\mathbf{h}}^{R}\Vert^2\Vert\bar{\mathbf{h}}^{(1)}\Vert^2)$ and circuit complexity of SIM implemented with D-RIS and BD-RIS.}
\label{fig:results}
\end{figure}

%%%%%%%%%%%%%%%%%%%%%%%%%%%%%%%%%%%%%%%%%%%%%%%%%%
\section{Conclusion}

We analyze SIM-aided systems by using multiport network theory and derive a physically consistent channel model.
To understand the role of the SIM on the derived model, we make two assumptions and obtain a simplified model in agreement with the model employed in previous works.
We compare SIM implemented with D-RIS and BD-RIS, showing that 1-layer SIM based on BD-RIS outperforms any SIM based on D-RIS.
Future research directions include studying the impact of the assumptions made on the channel model, optimizing SIM in more general systems, and prototyping of SIM.

%%%%%%%%%%%%%%%%%%%%%%%%%%%%%%%%%%%%%%%%%%%%%%%%%%
\section*{Appendix}

Denote as $\mathbf{a}_1=[a_{1,1},\ldots,a_{1,N_1}]^T$ and $\mathbf{b}_1=[b_{1,1},\ldots,b_{1,N_1}]^T$ the incident and reflected waves on the first $N_1$ ports of the $N_P$-port network; as $\mathbf{a}_2=[a_{2,1},\ldots,a_{2,N_2}]^T$ and $\mathbf{b}_2=[b_{2,1},\ldots,b_{2,N_2}]^T$ the incident and reflected waves on the last $N_2$ ports of the $N_P$-port network; and as $\mathbf{a}_3=[a_{3,1},\ldots,a_{3,N_3}]^T$ and $\mathbf{b}_3=[b_{3,1},\ldots,b_{3,N_3}]^T$ the reflected and incident waves on the last $N_3$ ports of the $N_Q$-port network.
$\mathbf{a}_2$ and $\mathbf{b}_2$ are also the reflected and incident waves on the first $N_2$ ports of the $N_Q$-port network.
Thus, we have
\begin{equation}
\begin{bmatrix}
\mathbf{b}_{1}\\
\mathbf{b}_{2}
\end{bmatrix}=
\begin{bmatrix}
\mathbf{P}_{11} & \mathbf{P}_{12}\\
\mathbf{P}_{21} & \mathbf{P}_{22}
\end{bmatrix}
\begin{bmatrix}
\mathbf{a}_{1}\\
\mathbf{a}_{2}
\end{bmatrix},
\:
\begin{bmatrix}
\mathbf{a}_{2}\\
\mathbf{a}_{3}
\end{bmatrix}=
\begin{bmatrix}
\mathbf{Q}_{11} & \mathbf{Q}_{12}\\
\mathbf{Q}_{21} & \mathbf{Q}_{22}
\end{bmatrix}
\begin{bmatrix}
\mathbf{b}_{2}\\
\mathbf{b}_{3}
\end{bmatrix},\label{eq:PQ}
\end{equation}
and our goal is to characterize $\mathbf{R}$ such that
\begin{equation}
\begin{bmatrix}
\mathbf{b}_{1}\\
\mathbf{a}_{3}
\end{bmatrix}=
\begin{bmatrix}
\mathbf{R}_{11} & \mathbf{R}_{12}\\
\mathbf{R}_{21} & \mathbf{R}_{22}
\end{bmatrix}
\begin{bmatrix}
\mathbf{a}_{1}\\
\mathbf{b}_{3}
\end{bmatrix},
\end{equation}
according to \cite[Chapter 4]{poz11}.
To this end, we first use \eqref{eq:PQ} to express $\mathbf{a}_{2}$ and $\mathbf{b}_{2}$ as functions of $\mathbf{a}_{1}$ and $\mathbf{b}_{3}$ as
\begin{equation}
\mathbf{a}_{2}
=\left(\mathbf{I}-\mathbf{Q}_{11}\mathbf{P}_{22}\right)^{-1}\mathbf{Q}_{11}\mathbf{P}_{21}\mathbf{a}_{1}+\left(\mathbf{I}-\mathbf{Q}_{11}\mathbf{P}_{22}\right)^{-1}\mathbf{Q}_{12}\mathbf{b}_{3},\label{eq:a2}
\end{equation}
\begin{equation}
\mathbf{b}_{2}
=\left(\mathbf{I}-\mathbf{P}_{22}\mathbf{Q}_{11}\right)^{-1}\mathbf{P}_{21}\mathbf{a}_{1}+\left(\mathbf{I}-\mathbf{P}_{22}\mathbf{Q}_{11}\right)^{-1}\mathbf{P}_{22}\mathbf{Q}_{12}\mathbf{b}_{3}.\label{eq:b2}
\end{equation}
Then, by substituting \eqref{eq:a2} into $\mathbf{b}_{1}=\mathbf{P}_{11}\mathbf{a}_{1}+\mathbf{P}_{12}\mathbf{a}_{2}$, which is derived from \eqref{eq:PQ}, we obtain
\begin{multline}
\mathbf{b}_{1}=\mathbf{P}_{11}\mathbf{a}_{1}+\mathbf{P}_{12}\left(\mathbf{I}-\mathbf{Q}_{11}\mathbf{P}_{22}\right)^{-1}\mathbf{Q}_{11}\mathbf{P}_{21}\mathbf{a}_{1}\\
+\mathbf{P}_{12}\left(\mathbf{I}-\mathbf{Q}_{11}\mathbf{P}_{22}\right)^{-1}\mathbf{Q}_{12}\mathbf{b}_{3},
\end{multline}
and by substituting \eqref{eq:b2} into $\mathbf{a}_{3}=\mathbf{Q}_{21}\mathbf{b}_{2}+\mathbf{Q}_{22}\mathbf{b}_{3}$, we obtain
\begin{multline}
\mathbf{a}_{3}=\mathbf{Q}_{21}\left(\mathbf{I}-\mathbf{P}_{22}\mathbf{Q}_{11}\right)^{-1}\mathbf{P}_{21}\mathbf{a}_{1}\\
+\mathbf{Q}_{21}\left(\mathbf{I}-\mathbf{P}_{22}\mathbf{Q}_{11}\right)^{-1}\mathbf{P}_{22}\mathbf{Q}_{12}\mathbf{b}_{3}+\mathbf{Q}_{22}\mathbf{b}_{3},
\end{multline}
proving the proposition.

\bibliographystyle{IEEEtran}
\bibliography{IEEEabrv,main}

\end{document}